\documentclass{llncs} 

\usepackage{amssymb}
\usepackage{textcomp}
\usepackage{amsmath}

\usepackage{amsthm}
\usepackage{algorithm}
\usepackage{algpseudocode}
\usepackage{comment}
\usepackage{graphicx}
\usepackage{color}
\newcommand{\ket}[1]{|#1\rangle}

\begin{document}

\title{Quantum-over-classical Advantage in Solving Multiplayer Games}
\author{Dmitry~Kravchenko$^{1}$ \and Kamil~Khadiev$^{2,3}$\and Danil~Serov$^{3}$ \and Ruslan Kapralov$^{3}$} 

\institute{Center for Quantum Computer Science, Faculty of Computing, University of Latvia, Riga, Latvia\and Smart Quantum Technologies Ltd., Kazan, Russia\and Kazan Federal University, Kazan, Russia \\ \email{kravchenko@gmail.com, kamilhadi@gmail.com, serovdanilru@gmail.com, kapralov$\_$ruslan@mail.ru} }

\maketitle

\begin{abstract}

We study the applicability of quantum algorithms in computational game theory
and generalize some results related to Subtraction games,
which are sometimes referred to as one-heap Nim games.

In quantum game theory, a subset of Subtraction games became the first explicitly defined
class of zero-sum combinatorial games with provable separation between quantum and classical
complexity of solving them.
For a narrower subset of Subtraction games, an exact quantum sublinear algorithm is known
that surpasses all deterministic algorithms for finding solutions with probability $1$.

Typically, both Nim and Subtraction games are defined for only two players.
We extend some known results to games for three or more players,
while maintaining the same classical and quantum complexities:
$\Theta\left(n^2\right)$ and $\tilde{O}\left(n^{1.5}\right)$ respectively.

\end{abstract}

\section{Introduction}

\emph{Quantum game theory} traditionally is being studied in the context of nonlocal properties of quantum particles, and usually stays apart from quantum computing and quantum algorithms.
In contrast, \emph{Quantum combinatorial game theory} is an approximately yearling branch of game theory, where quantum algorithms are applied for solving classical combinatorial games.

First explicit examples of combinatorial games with quantum-better-than-classical solving algorithms are some subsets of \emph{Subtraction games}.
\cite{kks2019} identify a specific subset of size $const^{n^2}$ of Subtraction games which are solvable by a quantum algorithm in time $O\left(n^{1.5}\log{n}\right)$ in bounded error setting.
\cite{hyzl2020} identify a smaller set of size $const^{\sqrt{n}\log{n}}$ of restricted Subtraction games which are solvable by an exact quantum algorithm in time $O\left(n^{1.5}\right)$.
Deterministic algorithms for both classes of games require $\Omega\left(n^2\right)$ steps for solving.
(Hereafter $n+1$ stands for the number of positions in a game, regardless which of the players has to make the next move.)

A \emph{Subtraction game} is similar to a canonical Nim game \cite{f2000} in several senses.
Nim is a notable game in game theory because it traditionally serves as a ``base case'' for Sprague-Grundy theorem \cite{s1935,g1939}, which establishes a deterministic upper bound for solving many combinatorial games.
Similarly, Subtraction games seem to become a good candidate for being a ``base sample'' for game-solving Quantum Dynamic Programming.
And like many games are known to be reducible to some Nim games, also many games on graphs can be reduced to the corresponding Subtraction games.
Finally, the rules of these games have very similar definitions.

The difference between these two games is that in a Subtraction game, the players deal with just one heap of stones, though with certain limitations imposed on the number of stones they can take from the heap.
The most common limitation for Subtraction games is defining a maximum for the number of stones to be taken away, and this kind of Subtraction game has very fast deterministic solutions in $\tilde{O}\left(1\right)$.
Here we study a much more general class of such limitations and thus a broader class of Subtraction games.

We investigate algorithms for \emph{solving} Subtraction games, that is determining the payoffs of all the players, assuming each of them to play optimally.
We exploit techniques similar to ones in \cite{kks2019} and \cite{hyzl2020} to establish upper bounds for the quantum complexity and asymptotes for the classical complexity.

The paper is organized in the following way. Section \ref{sec:defs} contains basic definitions. In Section \ref{sec:balanced-substr}, we present evaluations of classical complexity and quantum algorithms for solving a special class of games which we call balanced Subtraction games. Finally, in Section \ref{sec:restricted-substr} we analyze the complexity of solving the so called restricted Subtraction games.

\section{Definitions}\label{sec:defs}

\subsection{Subtraction Games}

In a play of a \emph{Subtraction game} players $1, \ldots, k$ sequentially remove some positive amounts of stones from a heap, with player $l$ being followed by player $\left(l\bmod{k}+1\right)$.

Let $n$ be the initial number of stones in the heap, and $\Gamma$ be a lower-triangular {\em binary} matrix of size $n \times n$, with rows numbered from $1$ to $n$ and columns numbered from $0$ to $n-1$.
A player which has to make the next move, can remove $j-i$ stones ($0 \le i < j \le n$) from the heap with exactly $j$ stones left iff $\Gamma_{ji}=1$.
In simple terms, $\Gamma_{ji}$ indicates the possibility for a player to receive position ``$j$ stones'' from the predecessor and pass position ``$i$ stones'' to the follower on the next turn.

If in some position a player, say player $l_0$, cannot make a legal move, then the play ends, and each player $l$ receives their payoff $\$\left(l-l_0+k\right)\bmod k$.
That is, player $l_0$ is the loser, the previous player is the major winner with payoff $\$\left(k-1\right)$, the previous-to-previous player gets $\$\left(k-2\right)$ and so on.
In order to become the major winner, a player has to take all the remaining stones, or to leave a number of stones $j$ such that no allowed moves would remain: $\sum\limits_{i=0}^{j-1} \Gamma_{ji} = 0$.

Obviously, the rules of a Subtraction game are fully determined by such matrix $\Gamma$, so hereafter we use letter $\Gamma$ to denote a corresponding game.
We also reserve the name $n$ to denote the initial number of stones. This number $n$ also corresponds to the dimension of the matrix $\Gamma$.
Note that there are only $n\left(n+1\right)/2$ meaningful bits in the matrix $\Gamma$, as a player cannot increase the number of stones in the heap or leave it as is: $\Gamma_{ji}=0$ for all $j \le i$.

Finally, we note that the selected payoff function is not a must.
It may be arbitrary, provided that each player has strict preferences over the set of all $k$ possible endings.
Otherwise, if some preferences are not strict, the concept of optimal behavior will not be well-defined, and the required assumption of optimal players will fail.

\subsection{Winning Function}
Let  ${\cal G}$ be a set of lower-triangular binary matrices of size $n \times n$, with rows numbered from $1$ to $n$ and columns numbered from $0$ to $n-1$.

We define a winning function $\textsc{Win}:{\cal G}\otimes {\cal N}_n \to \{\$w\}_{0\le w<k}$, such that $\textsc{Win}\left(\Gamma,j\right)=\$w$ iff a player gets payoff $\$w$ given position ``$j$ stones'' in game $\Gamma$, under assumption of optimal players:
$$
\textsc{Win}\left(\Gamma,j\right)=\begin{cases}
    \$0 & \text{if $j=0$,} \\
    \$0 & \text{if $\sum_i\Gamma_{ji}=0$,} \\
    \displaystyle\$\left(\max_{\substack{i : \Gamma_{ji}=1}} \textsc{Win}\left(\Gamma,i\right)-1\right) \bmod k & \text{otherwise.}
\end{cases}
$$
We also use notation $\textsc{Win}\left(\Gamma\right)=\textsc{Win}\left(\Gamma,n\right)$ for the value of game $\Gamma$.

\subsection{Properties of Subtraction Games}

In this work we stick to the conventional terminology of \cite[Section 2.3]{kks2019} and \cite[Section 2.1]{hyzl2020}, and use the following definitions.

We call a game $\Gamma$ \emph{losing} if the first player loses it assuming other players are optimal:
\begin{equation}\label{eq:losing}
\textsc{Win}\left(\Gamma\right)=\$0.
\end{equation}

We call a game $\Gamma$ \emph{balanced} if the values of $\textsc{Win}\left(\Gamma,j\right)$ are uniformly distributed over the set $\left\{\$w\right\}_{0 \le w < k}$:
\begin{equation}\label{eq:balanced}
\forall w: \Big|\#\big\{j:{\textsc{Win}\left(\Gamma,j\right)=\$w}\big\}_{1 \le j \le n} - \frac{n}{k}\Big| \le o\left(\frac{n}{k}\right).
\end{equation}

When considering a random balanced game we hereafter implicitly bear in mind the following procedure of picking a game:
\begin{enumerate}
\item Assign each position ``$j$ stones'' one of the values from $\left\{\$0,\$1,\ldots,\$\left(k-1\right)\right\}$ in the uniform fashion.
\item Assign position ``$0$ stones'' value $\$0$.
\item Initially assign $\Gamma = \left[0\right]_{ji}$.
\item For each position ``$j$ stones'' with value $\$w$ put $\Gamma_{ji}=1$ with probability $1/2$ whenever $i<j$ and position ``$i$ stones'' is assigned one of the values $\$\left(w+1\right) \bmod k,\quad \$w,\quad \$\left(w-1\right),\quad \ldots,\quad \$1$.
\item Additionally, for each position ``$j$ stones'' with value $\$w$ put $\Gamma_{ji}=1$ for one $i$ such that position ``$i$ stones'' is assigned value $\$\left(w+1\right) \bmod k$, whenever it was not already done in the previous step.
\item If the previous step failed because for some position ``$j$ stones'' it is not possible to find a position ``$i$ stones'' with the appropriate value, then start everything from the beginning.
\end{enumerate}
Should one feel that discarding in the last step essentially destroys the uniformity of $\textsc{Win}\left(\Gamma,j\right)$,
they can at the second step assign each position ``$w$ stones'', $0 \le w < k$, value $\$\left(k-w\right) \bmod k$.
This will make the last step obsolete, as no failure can occur, and will preserve the perfect uniformity.
Our further observations are valid for either kind of picking a random balanced Subtraction game.

Finally, we call a game $\Gamma$ \emph{restricted} if in each position at most one move is possible:
\begin{equation}\label{eq:restricted}
\forall j: \sum_i \Gamma_{ji} \le 1.
\end{equation}

\subsection{Computational Model}
To evaluate the complexity of a quantum algorithm, we use the standard form of the quantum query model. It is a generalization of the decision tree model of classical computation that is commonly used to lower bound the amount of time required by a computation.

    Let $f:D\rightarrow \{0,1\},D\subseteq \{0,1\}^n$ be an $n$ variable function we wish to compute on an input $x\in D$. We have an oracle access to the input $x$ --- it is realized by a specific unitary transformation usually defined as
    {%
    \relpenalty 100000
    \binoppenalty 100000
    $\ket{i}\ket{z}\ket{w}\rightarrow \ket{i}\ket{z \oplus x_i}\ket{w}$
    } where the $\ket{i}$ register indicates the index of the variable we are querying, $\ket{z}$ is the output register, and $\ket{w}$ is some auxiliary work-space. An algorithm in the query model consists of alternating applications of arbitrary unitaries independent of the input and the query unitary, and a measurement in the end. The smallest number of queries for an algorithm that outputs $f(x)$ with probability $\geq \frac{2}{3}$ on all $x$ is called the quantum query complexity of the function $f$ and is denoted by $Q(f)$. In this paper, as running time of an algorithm, we mean a number of queries to oracle. 
    
    More information on quantum computation and query model can be found in \cite{a2017,aazksw2019part1}.
    
    To distinguish ordinary deterministic and randomized complexities from the quantum complexity, they are traditionally called by one term \emph{classical complexity}.
    
\section{Balanced Subtraction Games}\label{sec:balanced-substr}

\subsection{Classical Query Complexity}\label{sec:classical-query-balanced}

In this subsection we limit our considerations with the number of players $k=3$, for the sake of simplicity:
\begin{itemize}
\item a player who cannot make a move at the end of a play gets $\$0$;
\item the previous player who managed to make the last move gets $\$2$;
\item the previous-to-previous player gets $\$1$.
\end{itemize}
All similar results also hold for arbitrary $k$ but require larger formulations, which in our mind are redundant for understanding. 

\begin{lemma}\label{sensitivitylemma}
Let $\Gamma$ be a losing balanced Subtraction game $\Gamma$ picked uniformly at random.
Let game $\Gamma'$ differ from $\Gamma$ in exactly one random bit of their binary representations: \sloppy{$\textsc{HammingDistance}\left(\Gamma,\Gamma'\right)=1$}.
\sloppypar{Then $\Pr\big[\textsc{Win}\left(\Gamma'\right)\neq\textsc{Win}\left(\Gamma\right)\big]\ge~1/18$}.
\end{lemma}
\begin{proof}

First, we mention the following three facts about balanced losing games.
\begin{itemize}
\item Balancedness \eqref{eq:balanced} of $\Gamma$ implies:
      \begin{equation}\label{eq:uniformity}
      \Pr_{\substack{j}}\big[\textsc{Win}\left(\Gamma,j\right)=\$1\big] =
      \Pr_{\substack{j}}\big[\textsc{Win}\left(\Gamma,j\right)=\$0\big] = \frac{1}{3}.
      \end{equation}
\item Losingness \eqref{eq:losing} of $\Gamma$ implies for each $j$, $0 < j < n$:
      \begin{equation}\label{eq:either_or}
      \text{either $\Gamma_{nj}=0$ or $\textsc{Win}\left(\Gamma,j\right)=\$1$, or both;}
      \end{equation}
      otherwise the first player would be able to take $n-j$ stones in the first turn and thus have a positive payoff.
\item Assuming $\Gamma$ to be picked at random, losingness \eqref{eq:losing} implies
      \begin{equation}\label{eq:conditional_one_half}
      \Pr\big[\Gamma_{nj}=1 ~\mid~ \textsc{Win}\left(\Gamma,j\right)=\$1\big] = \frac{1}{2},
      \end{equation}
      since possibility or impossibility of a (worst possible) move which leads to position ``$j$ stones'' with $\textsc{Win}\left(\Gamma,j\right)=\$1$ does not affect the value of any position.
\end{itemize}

Now let $\Gamma'_{ji} \ne \Gamma_{ji}$ for some pair of indices $j,i$ picked at random, s.t. $0\le i<j\le n$.
Then we are interested in the value of
$$
\Pr_{\substack{j,i}}\big[\textsc{Win}\left(\Gamma'\right)\neq \textsc{Win}\left(\Gamma\right)\big].
$$
Let us consider four possible cases for $j$ and $i$ to evaluate this probability.
Readers who only care about large values of $n$, are welcome to skip all but the last case, as the former ones are highly unlikely to happen for random $j$ and $i$.
\begin{enumerate}

  \item $j=n$, and $i=0$.

        Inversion of bit $\Gamma_{n0}$ changes the value of game $\Gamma$ with certainty, since $\Gamma$ is losing \eqref{eq:losing}, so and $\Gamma_{n0}=0$,
        but $\Gamma'$ with $\Gamma'_{n0}=1$ can be won by taking all $n$ stones in the first turn:
        $$ \Pr\big[\textsc{Win}\left(\Gamma'\right)\neq \textsc{Win}\left(\Gamma\right)\big] = 1 $$ 

  \item $j=n$, and $i>0$.

        $\eqref{eq:uniformity} \land \eqref{eq:either_or}$ implies that
        $\Pr\big[\textsc{Win}\left(\Gamma,i\right) \neq \$1 \land \Gamma_{ni}=0\big] = 1 - \frac{1}{3} = \frac{2}{3}$.
        Under this condition, inversion of bit $\Gamma_{ni}$ changes the value of game $\Gamma$ from $\$0$ to $\$\left(\textsc{Win}\left(\Gamma,i\right)-1\right) \bmod 3 > \$0$:
        $$ \Pr\big[\textsc{Win}\left(\Gamma'\right)\neq \textsc{Win}\left(\Gamma\right)\big] = \frac{2}{3} $$ 

  \item $j<n$, and $i=0$.

        $\eqref{eq:uniformity} \land \eqref{eq:either_or} \land \eqref{eq:conditional_one_half}$ implies that
        $\Pr\big[\textsc{Win}\left(\Gamma,j\right)=\$1 \land \Gamma_{nj}=1\big] = \frac{1}{3} \times \frac{1}{2} = \frac{1}{6}$
        Here $\textsc{Win}\left(\Gamma,j\right)=\$1$ means that $\Gamma_{j0}=0$, and inversion of this bit changes the value of $\textsc{Win}\left(\Gamma,j\right)$ from $\$1$ to $\$2$.
        And $\Gamma_{nj}=1$ means that the value of $\textsc{Win}\left(\Gamma\right)$ is at least $\left(\textsc{Win}\left(\Gamma,j\right)-\$1\right) \bmod k$.
        These two facts together mean that inversion of bit $\Gamma_{j0}$ changes the value of $\textsc{Win}\left(\Gamma,n\right)$ from $\$0$ to $\$1$. Therefore:
        $$ \Pr\big[\textsc{Win}\left(\Gamma'\right)\neq \textsc{Win}\left(\Gamma\right)\big] = \frac{1}{6} $$ 

  \item $0 < i < j < n$.

        $\eqref{eq:uniformity} \land \eqref{eq:either_or} \land \eqref{eq:conditional_one_half}$ implies that
        $\Pr\big[\textsc{Win}\left(\Gamma,j\right)=\$1 \land \textsc{Win}\left(\Gamma,i\right)=\$0 \land \Gamma_{nj}=1\big] = \frac{1}{3} \times \frac{1}{3} \times \frac{1}{2} = \frac{1}{18}$.
        Here $\textsc{Win}\left(\Gamma,j\right)=\$1 \land \textsc{Win}\left(\Gamma,i\right)=\$0$ means that $\Gamma_{ji}=0$, and inversion of this bit changes the value of $\textsc{Win}\left(\Gamma,j\right)$ from $\$1$ to $\$2$.
        And $\Gamma_{nj}=1$ means that the value of $\textsc{Win}\left(\Gamma\right)$ is at least $\left(\textsc{Win}\left(\Gamma,j\right)-\$1\right) \bmod k$.
        These two facts together mean that inversion of bit $\Gamma_{ji}$ changes the value of $\textsc{Win}\left(\Gamma,n\right)$ from $\$0$ to $\$1$. Therefore:
        $$ \Pr\big[\textsc{Win}\left(\Gamma'\right)\neq \textsc{Win}\left(\Gamma\right)\big] = \frac{1}{18} $$ 

\end{enumerate}
In either case we have that $\Pr\big[\textsc{Win}\left(\Gamma'\right)\neq \textsc{Win}\left(\Gamma\right)\big] \ge \frac{1}{18}$.
\end{proof}

\begin{theorem}\label{th:classic}
There is no deterministic or randomized algorithm for solving function $\textsc{Win}$ faster than in $\Omega\left(n^2\right)$ steps.
\end{theorem}

The proof of this theorem is next to the obvious implication of Lemma~\ref{sensitivitylemma} and is identical to \cite[Theorem~1]{kks2019}.
We refer to the aforementioned work for the details and for the remarks on the eligibility of this proof, which are also relevant here.

\subsection{Quantum Algorithm}\label{sec:quantum-query-balanced}

In this subsection we suggest a quantum algorithm for solving an arbitrary $k$-player Subtraction game.
We assume the reader to be familiar with the basics of quantum computing and, in particular, with Grover Search algorithm \cite{g1996,bbht1998}.
Among other problems, this algorithm is also applicable for searching in directed acyclic graphs (DAGs) \cite{ks2019,cormen2001}.
In this paper we apply it to Subtraction games which essentially are games on DAGs: if a game-representing binary $\Gamma$ is treated as an adjacency matrix of DAG $G=\left(V,E\right)$, then its set $V$ corresponds to $n$ positions of the game, and its set $E$ corresponds to all the legal moves.

The algorithm determines an optimal move in each possible position, thus providing a strong solution of a game.
But as we are interested in the value of $\textsc{Win}\left(\Gamma\right)$ only, we provide a solution which only returns the value of a game.
The algorithm searches for the maximum among directly accessible vertices $\textsc{Adj}\left[j\right] \stackrel{\text{def}}{=} \left\{i: \Gamma_{ji}\right\}$.
This algorithm has two important properties:
\begin{itemize}
\item its expected running time is $O\left(\sqrt{\deg{j}}\right)$, where $\deg{j} \stackrel{\text{def}}{=} \big|\textsc{Adj}\left[j\right]\big|$ is the number of vertices directly accessible from the vertex $j$;
\item it returns a vertex $i'$ with the maximal value of $\textsc{Win}$ with a constant probability (say $0.5$) if there exist one or more vertices with maximal values.
\end{itemize}
In Algorithm \ref{alg:quantumalgo}, we use D{\"u}rr-H{\o}yer Algorithm \cite{dh1996} for minimum search in the form of $\textsc{Grover\_Max}$ subroutine which returns maximum value among its arguments.
We store the search results in the array $w$ and reuse them in all the subsequent searches.

\begin{algorithm}
\caption{Quantum algorithm for solving a $k$-player Subtraction game $\Gamma$}\label{alg:quantumalgo}
\begin{algorithmic}
\State $w_0 \gets \$0$
\For{$j=1\ldots n$}\Comment{$O\left(n\right)$}
    \State $w_j \gets \$0$
    \For{$z=1\ldots 2 \cdot \log_2{n}$}\Comment{$O\left(\log{n}\right)$}
        \State $w_j \gets \max\left(w_j, \textsc{Grover\_Max} \left\{\$\left(w_i-1\right) \bmod k \mid i \in \textsc{Adj}\left[j\right]\right\}\right)$\Comment{$O\left(\sqrt{\deg{j}}\right)$}
    \EndFor
\EndFor
\State
\Return $w_n$
\end{algorithmic}
\end{algorithm}
\begin{theorem}\label{th:general-k-players}
\sloppy Algorithm~\ref{alg:quantumalgo} computes $\textsc{Win}\left(\Gamma\right)$ in expected running time $O\left(\sqrt{n\left|E\right|}\log n\right)$ and with error probability $\epsilon \lesssim 1/n$.
\end{theorem}
\begin{proof}
The correctness of the algorithm is obvious: each of the variables $w_j$ (for $j$ running from $1$ to $n$) is assigned a value $\textsc{Win}\left(\Gamma,j\right)$ according to the definition of the function $\textsc{Win}\left(\Gamma,j\right)$.

The time complexity follows from Cauchy-Bunyakovsky-Schwarz inequality:
$$ \sum_{j=1}^{n}{\sqrt{\deg{j}}} \le \sum_{j=1}^{n}{\sqrt{\mathbb{E}_j\left[\deg{j}\right]}}=\sum_{j=1}^{n}{\sqrt{\left|E\right|/n}}=\sqrt{n\left|E\right|}.$$

\sloppy{The probability of error in evaluating one particular $w_j$ is ${2^{-2\log_2 n}=1/{n^2}}$, so the probability of no error at all among evaluations of $w_1, \ldots, w_n$ is ${\left(1-1/{n^2}\right)^n \gtrsim 1-1/n}$.}

\end{proof}

We note that, for a random Subtraction game, the expected number of edges $\mathbb{E}\big[\left|E\right|\big] = \Theta\left(n^2\right)$, and then we conclude that, while the best classical algorithms require time $\Theta\left(n^2\right)$ to solve a Subtraction game, there exists a polynomially faster quantum algorithm which runs in time $O\left(n^{3/2}\log{n}\right)$.

\paragraph{The exact-time algorithm for a small number of players.}
If $k$ is a small constant, one can apply a quantum algorithm that works in \emph{exact} time $O(\sqrt{n|E|}\log n)$ (in contrast to Algorithm \ref{alg:quantumalgo} which has the same evaluation for the \emph{expected} running time). Algorithm \ref{alg:threeplayer} runs Grover's Search $k-1$ times instead of running one search for the maximum. At the $t$-th step the value $t$ is to be searched for among the values from the adjacent vertices, for $t$ running from $\$\left(k-1\right)$ down to $\$0$. Obviously, the first found value is equal to the maximal payoff available in the considered position ``$j$ stones''. Subroutine $\textsc{GROVER}_t$ in Algorithm \ref{alg:threeplayer} searches for the value $\$t$ among its arguments and returns $True$ with probability $0.5$ when there is such value, and $False$ otherwise. It has to be run $2\cdot \log_2 k \cdot \log_2 n$ times to amplify the probability of success in case if the arguments contain value $\$t$.

\begin{algorithm}
\caption{Quantum algorithm for solving a $k$-player Subtraction game $\Gamma$, where $k$ is a small constant}\label{alg:threeplayer}
\begin{algorithmic}
\State $w_0 \gets \$0$
\For{$j=1\ldots n$}\Comment{$O\left(n\right)$}
  \State $w_j\gets \$0$
  \State $t\gets \$\left(k-1\right)$
  \While{$t>\$0 ~~\land~~ w_j=\$0$}
    \For{$z=1\ldots 2 \cdot \log_2{n}\cdot \log_2 k$}\Comment{$O\left(\log{n}\right)$} 
      \If{$\textsc{Grover}_t \left\{\$\left(w_i-1\right) \bmod k \mid i \in \textsc{Adj}\left[j\right]\right\}$} \Comment{$O\left(\sqrt{\deg{j}}\right)$} 
        \State $w_j \gets t$
      \EndIf
    \EndFor
    \State $t \gets t-\$1$
  \EndWhile
\EndFor
\State \Return $w_n$
\end{algorithmic}
\end{algorithm}

\begin{theorem}
If $k$ is a small constant, Algorithm~\ref{alg:threeplayer} computes $\textsc{Win}\left(\Gamma\right)$ in exact running time $O\left(\sqrt{n\left|E\right|}\log n\right)$ and with error probability $\epsilon \lesssim 1/n$.
\end{theorem}

\section{Restricted Subtraction Games}\label{sec:restricted-substr}

\subsection{Concept of Solution}

Restricted Subtraction games are, in some sense, degenerated, as players essentially have no choice in either position.
Nevertheless, these games are of natural interest in terms of the computational complexity of Boolean functions.
Namely, they became the first functions formulated in terms of a game, that demonstrate polynomial separation between exact quantum query complexity and classical query complexity.

We note that the adaptive deterministic query complexity of $\textsc{Win}\left(\Gamma\right)$ for a restricted Subtraction game $\Gamma$ is $n$. The upper bound $n$ follows from a very simple analysis of Algorithm \ref{alg:restricted-simple}, the lower bound $n$ also is obvious.
\begin{algorithm}
\caption{Computing $\textsc{Win}\left(\Gamma\right)$ for a restricted Subtraction game $\Gamma$}\label{alg:restricted-simple}
\begin{algorithmic}
\State $w_n \gets \$0$
\State $j \gets n$
\For{$i=n-1,\ldots,0$}
  \If{$\Gamma_{ji}=1$}
    \State $j \gets i$
    \State $w_n \gets \$\left(w_n-1\right) \bmod k $
  \EndIf
\EndFor
\State \Return $w_n$
\end{algorithmic}
\end{algorithm}

Instead of computing $\textsc{Win}\left(\Gamma\right)$ for a restricted Subtraction game $\Gamma$, \cite{hyzl2020} aims for a more ambitious problem of solving all positions of a game, i.e. of finding vector $W=\big[\textsc{Win}\left(\Gamma,j\right)\big]_j$.

\subsection{Classical Query Complexity}
\begin{theorem}\label{th:restricted-classical}
Classical query complexity of computing $\big[\textsc{Win}\left(\Gamma,j\right)\big]_j$ for a $k$-player restricted Subtraction game $\Gamma$ is $\Theta\left(n^2\right)$.
\end{theorem}
The proof is identical to one of \cite[Theorem 1]{hyzl2020}, which was formulated for the two-player games, but is valid also for the multiplayer games.

\subsection{Quantum Algorithm}
To solve a restricted Subtraction game, we modify Algorithm \ref{alg:quantumalgo} according to the idea from \cite{hyzl2020}. We run \emph{exact} Grover's search for a non-zero element in each row. Exact Grover's search \cite{l2001} is a modification of Grover's algorithm, which returns the position of the non-zero element in a binary string with Hamming weight 1, or $False$ if its Hamming weight is 0. It cannot handle strings with a bigger Hamming weight, but for the promised input, it works in exact time $O\left(\sqrt{n}\right)$ for an $n$-bit binary string, and with no errors. Subroutine $\textsc{Exact\_Grover}$ of Algorithm \ref{alg:oneedge} refers to the exact Grover's search. In contrast to Algorithms \ref{alg:quantumalgo} and \ref{alg:threeplayer}, this subroutine has to be called just once as its result is exact and does not need to be amplified.

\begin{algorithm}
\caption{Quantum algorithm for solving a restricted $k$-player game $\Gamma$}\label{alg:oneedge}
\begin{algorithmic}
\State $w_0 \gets \$0$
\For{$j=1\ldots n$}\Comment{$O\left(n\right)$}
  \State $t \gets  \textsc{Exact\_Grover} \left\{ \textsc{Adj}\left[j\right]\right\}$\Comment{$O\left(\sqrt{\deg{j}}\right)$}
  \State $w_j \gets \$\left(w_t-1\right) \bmod k$
\EndFor
\State
\Return $w$
\end{algorithmic}
\end{algorithm}
\begin{theorem}
Algorithm~\ref{alg:oneedge} computes $\textsc{Win}\left(\Gamma\right)$ in exact running time $O\left(n^{1.5}\right)$ and with no error.
\end{theorem}
The theorem follows from the properties of exact Grover's search and the evaluation of $O(\sqrt{\deg{j}})$ as in the proof of Theorem \ref{th:general-k-players}.

\section{Conclusion}
Recent results in quantum game theory have stepped into the field of combinatorial games.
In this work we generalized several of these results for solving multiplayer combinatorial games.
In particular, we established several upper bounds for quantum query complexity, which generally correspond to the running time of a quantum algorithm.
We also derived several classical lower bounds for these problems.
We did not focus on the classical upper bounds, but they obviously coincide with the lower bounds, which can be shown just by describing the straightforward dynamic programming approach.

The polynomial separation between quantum and classical complexities was shown using different kinds of games and different concepts of solution,
but all of them engage Subtraction games for the demonstration of the power of quantum algorithms in combinatorial game theory.
Perhaps, one should expect better and more general bounds to emerge for the quantum complexity of solving Subtraction games.
Of course, we hope also for detecting other examples of games with quantum-smaller-than-classical complexity.

\subsection*{Acknowledgement}
The research is supported by PostDoc Latvia Program, and by the ERDF within the project 1.1.1.2/VIAA/1/16/099 ``Optimal quantum-entangled behavior under unknown circumstances''.
The reported study was funded by RFBR according to the research project No.19-37-80008.


\begin{thebibliography}{10}

\bibitem{aazksw2019part1}
F.~Ablayev, M.~Ablayev, H.~J.~Zhexue, K.~Khadiev, N.~Salikhova and D.~Wu 
\newblock {On quantum methods for machine learning problems part I: Quantum tools}
\newblock {\em Big Data Mining and Analytics}, 3(1):41--55, 2019.

\bibitem{a2017}
Andris Ambainis.
\newblock Understanding quantum algorithms via query complexity.
\newblock {\em arXiv preprint arXiv:1712.06349}, 2017.

\bibitem{bbht1998}
Michel Boyer, Gilles Brassard, Peter H{\o}yer, and Alain Tapp.
\newblock Tight bounds on quantum searching.
\newblock {\em Fortschritte der Physik}, 46(4-5):493--505, 1998.

\bibitem{cormen2001}
Thomas~H. Cormen, Charles~E. Leiserson, Ronald~L. Rivest, and Clifford Stein.
\newblock {\em Introduction to Algorithms-Secund Edition}.
\newblock McGraw-Hill, 2001.

\bibitem{dh1996}
D{\"u}rr, C. and H{\o}yer, P.
\newblock {\em A quantum algorithm for finding the minimum}.
\newblock {\em arXiv:quant-ph/9607014},
\newblock 1996.

\bibitem{f2000}
Thomas~S. Ferguson.
\newblock Game theory class notes for math 167, fall 2000, 2000.
\newblock
  https://www.cs.cmu.edu/afs/cs/academic/class/15859-f01/www/notes/comb.pdf.

\bibitem{g1996}
Lov~K.~Grover.
\newblock A fast quantum mechanical algorithm for database search.
\newblock In {\em Proceedings of the twenty-eighth annual ACM symposium on Theory of computing}, pages 212--219. ACM, 1996.

\bibitem{g1939}
P.~M. Grundy.
\newblock Mathematics and games.
\newblock {\em Eureka}, 2:6--8, 1939.

\bibitem{hyzl2020}
Y.~Huang, Z.~Ye, S.~Zheng and L.~Li.
\newblock An Exact Quantum Algorithm for a Restricted Subtraction Game.
\newblock {\em International Journal of Theoretical Physics}, 1--8, 2020.

\bibitem{ks2019}
Kamil Khadiev and Liliya Safina.
\newblock Quantum algorithm for dynamic programming approach for dags.
  applications for zhegalkin polynomial evaluation and some problems on dags.
\newblock In {\em Proceedings of Unconventional Computation and Natural Computation 2019}, LNCS. 2019.
  
\bibitem{kks2019}
 K.~Khadiev,  D.~Kravchenko and D.Serov.
\newblock On the Quantum and Classical Complexity of Solving Subtraction Games.
\newblock In {\em Proceedings of CSR 2019}, LNCS, 11532:228--236. 2019.  
  
\bibitem{l2001}
G.-L.~Long.
\newblock Grover algorithm with zero theoretical failure rate.
\newblock Physical Review A, 64(2):022307. 2001.

\bibitem{s1935}
R.~P. Sprague.
\newblock {\"U}ber mathematische kampfspiele.
\newblock {\em Tohoku Mathematical Journal}, 41:438--444, 1935.

\end{thebibliography}
\end{document}